\documentclass[journal]{IEEEtran}
\usepackage{verbatim}
\usepackage{algpseudocode}
\usepackage{algorithm}
\usepackage{mathtools}

\usepackage{etoolbox}
\apptocmd{\thebibliography}{\scriptsize}{}{}
\usepackage{cite}
\ifCLASSINFOpdf
   \usepackage[pdftex]{graphicx}
\else
   \usepackage[dvips]{graphicx}
\fi
\usepackage{amssymb}
\usepackage{amsthm}
\usepackage{mathtools}
\usepackage{bbm}
\ifCLASSOPTIONcompsoc
  \usepackage[caption=false,font=normalsize,labelfont=sf,textfont=sf]{subfig}
\else
  \usepackage[caption=false,font=footnotesize]{subfig}
\fi
\usepackage{color}
\providecommand{\mathbold}[1]{\mathbf{#1}}
\newcommand{\dpu}{\overline{\Delta p}}
\newcommand{\dpl}{\underline{\Delta p}}
\newcommand{\fu}{\overline{f}}
\newcommand{\fl}{\underline{f}}

\newtheorem{theorem}{Theorem}

\newtheorem{prop}[theorem]{Proposition}
\newtheorem{corollary}[theorem]{Corollary}

\newtheorem{definition}{Definition}
\newtheorem{remark}{Remark}

\usepackage{ifthen}{
\newboolean{showcomments}
\setboolean{showcomments}{true}

\newcommand{\dvij}[1]{\ifthenelse{\boolean{showcomments}}
{ \textcolor{red}{(Dj says:  #1)}}{}}
\newcommand{\adam}[1]{\ifthenelse{\boolean{showcomments}}
{ \textcolor{red}{(Adam says:  #1)}}{}}
\newcommand{\niangjun}[1]{\ifthenelse{\boolean{showcomments}}
{ \textcolor{red}{(Niangjun says:  #1)}}{}}
\newcommand{\navid}[1]{\ifthenelse{\boolean{showcomments}}
{ \textcolor{blue}{(Navid says:  #1)}}{}}

\newcommand{\note}[1]{\ifthenelse{\boolean{showcomments}}
{ \textcolor{red}{(#1)}}{}}
\newcommand{\todo}[1]{\ifthenelse{\boolean{showcomments}}
{ \textcolor{red}{(To do: #1)}}{}}

\newcommand{\addcites}[0]{\ifthenelse{\boolean{showcomments}}
{ \textcolor{red}{(add cite(s))}}{}}
\newcommand{\addcite}[0]{\ifthenelse{\boolean{showcomments}}
{ \textcolor{red}{(add cite(s))}}{}}

\newcommand{\hide}[1]{}

\renewcommand{\phi}{\varphi}

\providecommand{\mathbbm}{\mathbb} 

\newcommand{\mtx}[1]{\mathbold{#1}}

\newcommand{\cpi}{c_{\pi(i)}}

\DeclarePairedDelimiter\floor{\lfloor}{\rfloor}

\renewcommand{\mathbf}{{}}
\renewcommand{\boldsymbol}{{}}

\begin{document}

\newtheorem{lemma}[theorem]{Lemma}
\newtheorem{proposition}[theorem]{Proposition}
\title{Opportunities for Price Manipulation by Aggregators in Electricity Markets}
\author{Navid~Azizan~Ruhi,
        Krishnamurthy~Dvijotham,
        Niangjun~Chen,
        and~Adam~Wierman\vspace*{-13pt}
\thanks{The authors are with the Department of Computing and Mathematical Sciences, California Institute of Technology, Pasadena,
CA, 91125 USA e-mail: \{azizan,dvij,ncchen,adamw\}@caltech.edu}
\thanks{Manuscript received Month DD, YYYY; revised Month DD, YYYY.}}

\maketitle

\begin{abstract}
Aggregators are playing an increasingly crucial role in the integration of renewable generation in power systems. However, the intermittent nature of renewable generation makes market interactions of aggregators difficult to monitor and regulate, raising concerns about potential market manipulation by aggregators.  In this paper, we study this issue by quantifying the profit an aggregator can obtain through strategic curtailment of generation in an electricity market. We show that, while the problem of maximizing the benefit from curtailment is hard in general, efficient algorithms exist when the topology of the network is radial (acyclic).  Further, we highlight that significant increases in profit are possible via strategic curtailment in practical settings. 
\end{abstract}

\begin{IEEEkeywords}
Aggregators, renewables, optimal curtailment, market power, locational marginal price (LMP).
\end{IEEEkeywords}

\IEEEpeerreviewmaketitle

\section{Introduction}

Increasing the penetration of distributed, renewable energy resources into the electricity grid is a crucial part of building a sustainable energy landscape.  To date, the entities that have been most successful at promoting and facilitating the adoption of renewable resources have been \emph{aggregators}, e.g. as SolarCity, Tesla, Enphase, Sunnova, SunPower, ChargePoint \cite{bullis2012solarcity,john2014solarcity,Wesoff2016earnings}. These aggregators install and manage rooftop solar installations as well as household energy storage devices and electric vehicle charging systems. Some have fleets with upwards of 800 MW distributed energy resources \cite{SEIA,solarpowerworld}, and the market is expected to triple in size by 2020 \cite{landscape, outlook}.

Aggregators play a variety of important roles in the construction of a sustainable grid.  First, and foremost, they are on the front lines of the battle to promote installation of rooftop solar and household energy storage, pushing for wide-spread adoption of distributed energy resources by households and businesses.  Second, and just as importantly, they provide a single interface point where utilities and Independent System Operators (ISOs) can interact with a fleet of distributed energy resources across the network in order to obtain a variety of services, from renewable generation capacity to demand response.  This service is crucial for enabling system operators to manage the challenges that result from unpredictable, intermittent renewable generation, e.g., wind and solar.

However, in addition to the benefits they provide, aggregators also create new challenges -- both from the perspective of the aggregator and the perspective of the system operator.  On the side of the aggregator, the management of a geographically diverse fleet of distributed energy resources is a difficult algorithmic challenge.  On the side of the operator, the participation of aggregators in electricity markets presents unique challenges in terms of monitoring and mitigating the potential of exercising market power.  In particular, unlike traditional generation resources, the ISO cannot verify the availability of the generation resources of aggregators. While the repair schedule of a conventional generator can be made public, the downtime of a solar generation plant and the times when solar generation is not available cannot be scheduled or verified after the fact. Thus, aggregators have the ability to strategically curtail generation resources without the knowledge of the ISO, and this potentially creates significant opportunities for them to manipulate prices.

These issues are particularly salient given current proposals for distribution systems. Distribution systems (which are typically radial networks) are heavily impacted by the introduction of distributed energy resources. As a result, there are a variety of current proposals to start distribution-level power markets (see, for example \cite{NYREV} \cite{FERCOrder}), operated by Distribution System Operators (DSOs). A future grid may even involve a hierarchy of system operators dealing with progressively larger areas, net load and net generation. In such a scenario, aggregators could end up having a significant proportion of the market share, and such markets may be particularly vulnerable to strategic bidding practices of the aggregators. Thus, understanding the potential for these aggregators to exercise market power is of great importance, so that regulatory authorities can take appropriate steps to mitigate it as needed.

\subsection{Summary of Contributions}

This paper addresses both the algorithmic challenge of managing an aggregator and the economic challenge of measuring the potential for an aggregator to manipulate prices.  Specifically, this work provides a new algorithmic framework for managing the participation of an aggregator in electricity markets, and uses this framework to evaluate the potential for aggregators to exercise market power.  To those ends, the paper makes three main contributions.

First, we introduce a new model for studying the management of an aggregator.  Concretely, we introduce a model of an aggregator in a two-stage market, where the first stage is the \textit{ex-ante} (or day-ahead) market that decides the generation schedule based on the solution to a security-constrained economic dispatch problem and the second stage is the \textit{ex-post} (or real-time) market to perform fine adjustment via Locational Marginal Prices (LMPs) based on updated information.

Second, we provide a novel algorithm for managing the participation of an aggregator in the two stage market.  The problem is NP-hard in general and is a bilevel quadratic program, which are notoriously difficult in practice.  However, we develop an efficient algorithm that can be used by the aggregators in radial networks to approximate the optimal curtailment strategy and maximize their profit (Section~\ref{sec: DP}). Note that the algorithm is not just relevant for aggregators; it can also be used by the operator to asses the potential for strategic curtailment. The key insight in the algorithm is that the optimization problem can be decomposed into ``local'' pieces and be solved approximately using a dynamic programming over the graph.

Third, we quantify opportunities for price manipulation (via strategic curtailment) by aggregators in the two stage market (Section~\ref{sec:aggregator}). Our results highlight that, in practical scenarios, strategic curtailment can have a significant impact on prices, and yield much higher profits for the aggregators. In particular, the prices can be impacted up to a few tens of \$/MWh in some cases, and there is often more than 25\% higher profit, even with curtailments limited to 1\%. Further, our results expose a connection between the profit achievable via curtailment and a new market power measure introduced in \cite{yu2014wind}.

\subsection{Related Work}

This paper connects to, and builds on, work in four related areas: (i) Quantifying and mitigating market power, (ii) Cyber-attacks in the grid, (iii) Algorithms for managing distributed energy resources, and (iv) Algorithms for bilevel programs.  

\paragraph*{Quantifying market power in electricity markets} There  is a large volume of literature that focuses on identifying and measuring market power for generators in an electricity market, see \cite{twomey2015} for a recent survey.

Early works on market power analysis emerged from microeconomic theory suggest measures that ignore transmission constraints. For example, \cite{bushnell1999} introduced the \textit{pivotal supplier index} (PSI), which is a binary value indicating whether the capacity of a generator is larger than the supply surplus, \cite{sheffrin2002} later refined PSI by proposing \textit{residential supply index} (RSI). RSI is used by the California ISO to assure price competitiveness \cite{caiso-practice}. The electricity reliability council of Texas uses the \textit{element competitiveness index} (ECI) \cite{ercot-practice}, which is based on the \textit{Herfindahl-Hirschmann index} (HHI) \cite{hhi-reference}.

Market power measures considering transmission constraints have emerged more recently.  Some examples include,  e.g., \cite{scheffman1987geographic, borenstein1995market,oren1997economic,cardell1997market, xu2007transmission}, and \cite{xu2002transmission}.  Interested readers can refer to \cite{bose2014unifying}, which proposes a functional measure that unifies the structural indices measuring market power on a transmission constrained network in the previous work.

In contrast to the large literature discussed above, the literature focused on market power of renewable generation producers is limited.
Existing works such as \cite{yu2014wind} and \cite{twomey2010wind} study market power of wind power producers ignoring transmission constraints.  The key differentiator of the work in this paper is that the use of the Locational Marginal Price (LMP) framework, which is standard practice in the electricity market \cite{ott2003experience, zheng2006ex}, allows this work to offer insight about market power of aggregators when transmission capacity is limited.

\paragraph*{Cyber-attacks in the grid} The model and analysis in this paper is also strongly connected to the cyber security research community, which has studied how and when a malicious party can manipulate the spot price in electricity markets by compromising the state measurement of the power grid via false data injection, e.g., see \cite{bi2013false,6840294,xie2010false,xie2011integrity,liu2011false}.

In particular, \cite{xie2010false, xie2011integrity} shows that if a malicious party can corrupt sensor data, then it can create an arbitrage opportunity.  Further, \cite{bi2013false} shows that such attacks can impact both the real time spot price and future prices by causing line congestions.

In this paper, we do not allow aggregators to corrupt the state measurements of the power system, rather we consider a perfectly legal approach for price manipulation: strategic curtailment.  However, strategic curtailment in the \textit{ex-post} market can gain extra profit to the detriment of the power system, which is a similar mechanism to those highlighted in cyber attack literature.  Technically, the work in this paper makes significant algorithmic contributions to the cyber-attack literature.  In particular, the papers mentioned above focus on algorithmic heuristics and do not provide formal guarantees.  In contrast, our work presents a polynomial time algorithm that provably maximizes the profit of the aggregator.

\paragraph*{Algorithms for managing distributed energy resources} There has been much work studying optimal strategies for managing demand response and distributed generation resources to offer regulation services to the power grid.  This work covers a variety of contexts.  For example, researchers have studied frequency regulation\cite{lin2015experimental}\cite{hao2013ancillary} and voltage regulation (or volt-VAR control) \cite{zhang2015optimal}\cite{ModelFreeVoltage}. A separate line of work has been work on designing incentives to encourage distributed resources to provide services to the power grid \cite{negrete2012markets}\cite{chen2014individual}. However, the current paper is distinct from all the work above in that we study strategic behavior by an aggregator of distributed resources.  Prior work does not model the strategic manipulation of prices by the aggregator.

\paragraph*{Algorithms for bilevel programs} The optimization problem that the strategic aggregator solves is a bilevel program, since the objective (aggregator's profit) depends on the locational prices (LMPs). The LMPs are constrained to be equal to optimal dual variables arising from economic dispatch-based market clearing procedure. These types of problems have been extensively studied in the literature, and fall under the class of Mathematical Programs with Equilibrium Constraints (MPECs) \cite{ferris2002mathematical}. Even if the optimization problems at the two levels is linear, the problem is known to be NP-hard \cite{ben1990computational}. Global optimization algorithms \cite{gumucs2001global} can be used to solve this problems to arbitrary accuracy (compute a lower bound on the objective within a specified tolerance of the global optimum). However, these algorithms use a spatial branch and bound strategy, and can take exponential time in general. In contrast, solvers like PATH \cite{dirkse1995path}, while practically efficient for many problems, are only guaranteed to find a local optimum. In this paper, we show that for tree-structured networks (distribution networks), an $\epsilon$-approximation of the global optimum can be computed in time linear in the size of the network and polynomial in $\frac{1}{\epsilon}$. 
\section{System Model}
In this section we define the power system model that serves as the basis for the paper and describe how we model the way the Independent System Operator (ISO) computes the Locational Marginal Prices (LMPs). Locational marginal pricing is adopted by the majority of power markets in the Unites States \cite{zheng2006ex}, and our model is meant to mimic the operation of two-stage markets like ISO New England, PJM Interconnection, and Midcontinent ISO, that use ex-post pricing strategy for correcting the ex-ante prices \cite{ott2003experience, zheng2006ex}.

We consider a power system with $n$ nodes (buses) and $t$ transmission lines. The generation and load at node $i$ are denoted by $p_i$ and $d_i$ respectively, with $\mathbf{p}=\begin{bmatrix}p_1,\dots,p_n\end{bmatrix}^T$ and $\mathbf{d}=\begin{bmatrix}d_1,\dots,d_n\end{bmatrix}^T$. We use $[n]$ to denote the set of buses $\{1,\ldots,n\}$. 

The focus of this paper is on the behavior of an aggregator, which owns generation capacity, possibly at multiple nodes.  We assume that the aggregator has the ability to curtail generation without penalty, e.g., by curtailing the amount of wind/solar generation or by not calling on demand response opportunities.  Let $N_a\subseteq [n]$ be the nodes where the aggregator has generation and denote its share of generation at node $i\in N_a$ by $p^a_i$ (out of $p_i$). The curtailment of generation at this node is denoted by $\alpha_i$, where $0\leq \alpha_i\leq p^a_i$. We define our model for the decision making process of the aggregator with respect to curtailment in Section \ref{sec:aggregator}.

Together, the net generation delivered to the grid is represented by $\mathbf{p}-\boldsymbol{\alpha}$, where $\alpha_j=0 \ \forall j\not\in N_a$. The flow of lines is denoted by $\mathbf{f}=\begin{bmatrix}f_1,\dots,f_t\end{bmatrix}^T$, where $f_l$ represents the flow of line $l$: $
\mathbf{f}=\mathbf{G (p-\boldsymbol{\alpha}-d)},
$
where $\mathbf{G}\in \mathbb{R}^{t\times n}$ is the matrix of generation shift factors \cite{shift_factor}. We also define $\mathbf{B}\in \mathbb{R}^{n\times t}$ as the link-to-node incidence matrix that transforms line flows back to the net injections as $
\mathbf{p-\boldsymbol{\alpha}-d}=\mathbf{Bf}.$

\subsection{Real-Time Market Price}\label{Real-Time Market Price}

At the end of each dispatch interval, in real time, the ISO obtains the current values of generation, demands and flows from the state estimator. Based on this information, it solves a constrained optimization problem for market clearing. The objective of the optimization is to minimize the total cost of the network, based on the current state of the system. The ex-post LMPs are announced as a function of the optimal Lagrange multipliers of this optimization. Mathematically, the following program has to be solved.
\begin{subequations}\label{ISO_opt}
\begin{align}
& \underset{\mathbf{f}}{\text{minimize}} & & \mathbf{c^TBf} \\
& \text{subject to} \notag\\
& \boldsymbol{\lambda}^-,\boldsymbol{\lambda}^+: && \mathbf{\dpl} \leq \mathbf{Bf-p+\boldsymbol{\alpha}+d} \leq \mathbf{\dpu}\label{const: p}\\
& \boldsymbol{\mu}^-,\boldsymbol{\mu}^+: && \mathbf{\fl} \leq \mathbf{f} \leq \mathbf{\fu}\label{const: f}\\
& \boldsymbol{\nu}: && \mathbf{f}\in \mathrm{range} (\mathbf{G})\label{const: range}
\end{align}
\end{subequations}
In the above, $c_i$ is the offer price for the generator $i$. $f_l$ is the desired flow of line $l$, and $\mathbf{Bf}=\mathbf{p+\Delta p-\boldsymbol{\alpha}-d}$, where $\Delta p_i$ is the desired amount of change in the generation of node $i$. Constraint \eqref{const: p} enforces the upper and lower limits on the change of generations, and constraint \eqref{const: f} keeps the flows within the line limits. In practice, the generation change limits are often set to be constant values of $\dpl_i = -2$ and $\dpu_i = 0.1$, for all $i$, \cite{patton20142013}.
The last constraint ensures that $f_l$ are valid flows, i.e. $\mathbf{f=G\tilde{p}}$ for some generation $\mathbf{\tilde{p}}$.
Variables $\boldsymbol{\lambda}^-,\boldsymbol{\lambda}^+\in\mathbb{R}^n_+$, $\boldsymbol{\mu}^-,\boldsymbol{\mu}^+\in\mathbb{R}^t_+$ and $\boldsymbol{\nu}\in\mathbb{R}^{t-\mathrm{rank}(G)}$ denote the Lagrange multipliers (dual variables) corresponding to constraints \eqref{const: p}, \eqref{const: f} and \eqref{const: range}.

Note that the ISO does not physically redispatch the generations, and the optimal values of the above program are just the desired values. In fact, by announcing the (ex-post) LMPs, the ISO provides incentives for the generators to adjust their generation according to the goals of the ISO \cite{zheng2006ex}. 

\begin{definition}
The \textbf{ex-post locational marginal price (LMP) }of node $i$ at curtailment level of $\boldsymbol{\alpha}$, denoted by $\lambda_i(\boldsymbol{\alpha})$, is 
\begin{equation}
\lambda_i(\boldsymbol{\alpha}) 
= c_i + \lambda_i^{+}(\alpha) - \lambda_i^{-}(\alpha). \label{LMP2}
\end{equation}
\end{definition}
We assume that there is a unique optimal primal-dual pair, and therefore the LMPs are unique. In general, there are several ways that ISOs ensure such condition, for instance by adding a small quadratic term to the objective.

\section{The Market Behavior of the Aggregator}
\label{sec:aggregator}

The key feature of our model is the behavior of the aggregator. As mentioned before, aggregators have generation resources at multiple locations in the network and can often curtail generation resources without the knowledge of the ISO. Of course, such curtailment may not be in the best interest of the aggregator, since it means offering less generation to the market. But, if through curtailment, prices can be impacted, then the aggregator may be able to receive higher prices for the generation offered or make money through arbitrage of the price differential.

In this section, we develop a model for a strategic aggregator, which we then use in order to understand in what situations curtailment may be beneficial.

\subsection{Preliminaries}

To quantify the profit that the aggregator makes due to the curtailment, let us take a look at the total revenue in different production levels.
\begin{definition}


We define the \textbf{curtailment profit (CP)} as the change in profit of the aggregator as a result of curtailment:
\begin{equation}\label{CP}
\gamma(\boldsymbol{\alpha}) = \sum_{i\in N_a}\left(\lambda_i(\boldsymbol{\alpha})\cdot (p^a_i-\alpha_i)-\lambda_i(0)\cdot p^a_i\right)
\end{equation}
\end{definition}

Note that the curtailment profit can be positive or negative in general. We say a curtailment level $\boldsymbol{\alpha}>0$ is profitable if $\gamma(\boldsymbol{\alpha})$ is strictly positive.



The curtailment profit is important for understanding when it is beneficial for the aggregators to curtail. Note that we are not concerned about the cost of generation here, as renewables have zero marginal cost.  However, if there is a cost for generation, then that results in an additional profit during curtailment, which makes strategic curtailment more likely.



\subsection{A Profit-Maximizing Aggregator}

A natural model for a strategic aggregator is one that maximizes curtailment profit subject to LMPs and curtailment constraints. Since LMPs are the solution to an optimization problem themselves, the aggregator's problem is a bilevel optimization problem. In order to be able to express this optimization in an explicit form, let us first write the KKT conditions of the program \eqref{ISO_opt}. \vspace{.05in}

\begin{subequations}\label{KKT}
\noindent Primal feasibility:
\begin{equation}
\mathbf{\dpl} \leq \mathbf{Bf-p+\boldsymbol{\alpha}+d} \leq \mathbf{\dpu} \label{eq:DCPF}
\end{equation}
\begin{equation}
\mathbf{\fl} \leq \mathbf{f} \leq \mathbf{\fu}
\end{equation}
\begin{equation}
\mathbf{Hf=0}
\label{eqn: Hf=0}
\end{equation}
Dual feasibility:
\begin{equation}
\boldsymbol{\lambda^-},\boldsymbol{\lambda^+},\boldsymbol{\mu^{-}},\boldsymbol{\mu^{+}} \geq 0
\end{equation}
Complementary slackness:
\begin{equation}
\lambda_i^+ ((Bf)_i-p_i+\alpha_i+d_i-\dpu_i) = 0, \quad i = 1, \dots, n
\end{equation}
\begin{equation}
\lambda_i^- (\dpl_i-(Bf)_i+p_i-\alpha_i-d_i) = 0, \quad i = 1, \dots, n
\end{equation}
\begin{equation}
\mu_l^+ (f_l-\fu_l) = 0, \quad l = 1, \dots, t
\end{equation}
\begin{equation}
\mu_l^- (\fl_l-f_l) = 0, \quad l = 1, \dots, t
\end{equation}
Stationarity:
\begin{equation}\label{eqn: stationarity}
\mathbf{B}^T(\boldsymbol{c+\lambda^+-\lambda^-})+\boldsymbol{\mu}^+-\boldsymbol{\mu}^-+\mathbf{H}^T\boldsymbol{\nu} = \mathbf{0}
\end{equation}
\end{subequations}
Here $\mathbf{H}\in\mathbb{R}^{(t-\mathrm{rank}(G))\times t}$, and the range of $\mathbf{G}$ is the nullspace of $\mathbf{H}$.

Using the KKT conditions derived above, the aggregator's problem can be formulated as follows.
\begin{subequations}\label{agg_opt}
\begin{align}
\gamma^*=& \underset{\boldsymbol{\alpha},\mathbf{f},\boldsymbol{\lambda}^-,\boldsymbol{\lambda}^+,\boldsymbol{\mu}^-,\boldsymbol{\mu}^+,\boldsymbol{\nu}}{\text{maximize}} \quad \gamma(\boldsymbol{\alpha})\label{objective}\\
&  \text{subject to} \quad\notag\\
& 0\leq \alpha_i\leq p^a_i, \quad i\in N_a\label{aggnodes}\\
& \alpha_j=0, \quad j\not\in N_a \label{othernodes}\\
& \eqref{KKT}\label{constraint_lmp}
\end{align}
\end{subequations}

The objective \eqref{objective} is the curtailment profit defined in \eqref{CP}. Constraints \eqref{aggnodes} and \eqref{othernodes} indicate that the aggregator can only curtails generation at its own nodes, and the amount of curtailment cannot exceed the amount of available generation to it. Constraint \eqref{constraint_lmp}, which is the KKT conditions, enforces the locational marginal pricing adopted by the ISO. 
Note that if one assumes that curtailments of a higher amount than a limit $\tau$ can be detected by the ISO, then we can simply replace $p_i(a)$ in \eqref{aggnodes} by $\tau_i$.

An important note about this problem is that we have assumed the aggregator has complete knowledge of the network topology ($\mathbf{G}$), and state estimates ($\mathbf{p}$ and $\mathbf{d}$).  This is, perhaps, optimistic; however one would hope that the market design is such that aggregators do not have profitable manipulations even with such knowledge.  The results in this paper indicate that this is not the case.

\subsection{Connections between Curtailment Profit and Market Power}

While our setup so far seems divorced from the notion of market power, it turns out that there is a fundamental relationship between market power and the notion of curtailment profit introduced above.

As mentioned earlier, there has been significant work on market power in electricity markets, but work is only beginning to emerge on the market power of renewable generation producers.  One important work from this literature is \cite{yu2014wind}, and the following is the proposed notion of market power from that work.

\begin{definition}
For $\alpha_i^*\geq 0$, the \textbf{market power} (ability) of the aggregator is defined as
\begin{equation}
\eta_i = \left(\frac{\lambda_i(\alpha^*)-\lambda_i(0)}{\lambda_i(0)}\right)/\left(\frac{\alpha_i^*}{p_i^a}\right)
\end{equation}
\end{definition}

In this definition the value of $\eta_i$ captures the ability of the generator/aggregator to exercise market power. Intuitively, in a market with high value of $\eta_i$, the aggregator can significantly increase the price by curtailing a small amount of generation.

Interestingly, the optimal curtailment profit is closely related to this notion of market power.  We summarize the relationship in the following proposition, which is proven in Appendix~\ref{sec:marketpowerproof}.

\begin{proposition}\label{thm: marketpower}
If the curtailment profit $\gamma$ is positive then the market power $\eta_i>1$. Furthermore, the larger the curtailment profit is, the higher is the market power.
\end{proposition}

This proposition highlights that the notion of market power in \cite{yu2014wind} is consistent with an aggregator seeking to maximize their curtailment profit, and higher curtailment profit corresponds to more market power.

\section{Optimizing Curtailment Profit}\label{sec: DP}

\newcommand{\br}[1]{\left({#1}\right)}
\newcommand{\delp}{\Delta p}
\newcommand{\powb}[2]{{\br{#1}}^{{#2}}}
\newcommand{\Xcal}{\mathcal{X}}
\newcommand{\el}{\ell}


The aggregator's profit maximization problem is challenging to analyze, as one would expect given its bilevel form.  In fact, bilevel linear programming is NP-hard to approximate up to any constant multiplicative factor in general \cite{dempe2015bilevel}. Furthermore, the objective of the program \eqref{agg_opt} is quadratic (bilinear) in the variables, rather than linear. 
This combination of difficulties means that we cannot hope to provide a complete analytic characterization of the behavior of a profit maximizing aggregator.

In this section, we begin with the case of a single-bus aggregator and build to the case of general multi-bus aggregators in acyclic networks.  For the single-bus aggregator, the optimal curtailment can be found exactly, in polynomial time.
For the general case, we cannot provide an exact algorithm, but we do provide a practical approximation algorithm for general multi-bus aggregators in acyclic networks (e.g. distribution networks).



\subsection{An Exact Algorithm for a Single-Node Aggregator}

Even in the simplest case, when the aggregator has only a single node, i.e. its entire generation is located in a single bus, it is not trivial how to solve the aggregator's profit maximization problem.

The first step toward solving the problem is already difficult.  In particular, in order to understand the effect of curtailment on the profit, we first need to understand how does curtailment impact the prices -- an impact which is not monotonic in general. However, although LMPs are not monotonic in general, it turns out that in single-bus curtailment, the LMP is indeed monotonic with respect to the curtailment.  The proof of the following lemma is in Appendix~\ref{sec:monoproof}.

\begin{lemma}\label{lem: mono}
The LMP of any bus $i$ is monotonically increasing with respect to the curtailment at that bus. That is
$$\lambda_i(\boldsymbol{\alpha}')\geq \lambda_i(\boldsymbol{\alpha})$$
if $\alpha_i'>\alpha_i$, and $\alpha_j'=\alpha_j$ for all $j=[n]\backslash \{i\}$.
\end{lemma}

A consequence of the above lemma is that the price $\lambda_i$ is a monotonically increasing staircase function of $\alpha_i$, for any bus $i$, as depicted in Fig. \ref{plot_monotonic}. Further, the binding constraints of \eqref{ISO_opt} do not change in the intervals, and thus the dual variables remain the same, i.e., the LMPs remain constant. When a constraint becomes binding/non-binding, the LMP jumps to the next level.

In Fig. \ref{plot_monotonic}, the blue and red shaded areas are profit at the normal condition and at the curtailment, respectively. The difference between the two areas the curtailment profit. In particular, if the red area is larger than the blue one, the aggregator is able to earn a positive curtailment profit on bus $i$. The optimal curtailment $\alpha_i^*$ also happens where the red area is maximized.
It should be clear that the optimal curtailment always happens at the verge of a price change, not in the middle of a constant interval (otherwise it can be increased by curtailing less).

\begin{figure}[tpb]
\centering
\includegraphics[width=0.6\columnwidth]{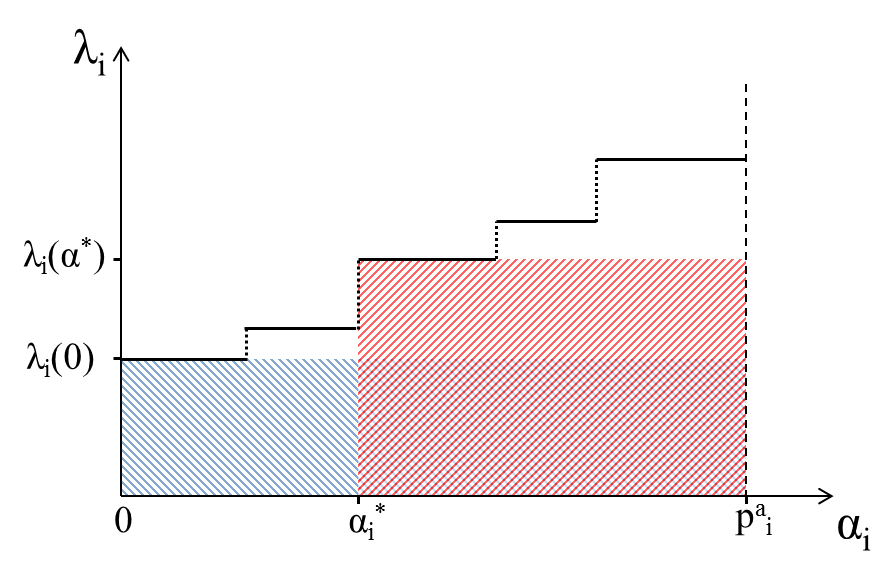}
\caption{The LMP at bus $i$ as a function of curtailed generation at that bus. Shaded areas indicate the aggregator's revenue at the normal condition and at the aggreation.}
\label{plot_monotonic}
\end{figure}

Given the knowledge of the network and state estimates, it is possible to find the jump points (i.e. where the binding constraints change) and evaluate them for profitability. Therefore, if there are not too many jumps, an exhaustive search over the jump points can yield the optimal curtailment. Based on this observation, we have the following theorem, which is proven in Appendix~\ref{sec:exactproof}.

\begin{theorem}
The exact optimal curtailment for an aggregator with a single bus, in a network with $t$ lines, can be found by an algorithm with running time $O(t^{3.373})$.
\label{thm: single_bus_opt}
\end{theorem}

Clearly, this approach does not extend to large multi-bus aggregators.  The following sections use a different, more sophisticated algorithmic approach for that setting.

\subsection{An Approximation Algorithm for Multi-Bus Aggregators in Radial Networks}

In this section, we show that the aggregator profit maximization problem, while hard in general, can be solved in an approximate sense to determine an approximately-feasible approximately-optimal curtailment strategy in polynomial time using an approach based on dynamic programming. In particular, we show that an $\epsilon$-approximation of the optimal curtailment profit can be obtained using an algorithm with running time that is linear in the size of the network and polynomial in $\frac{1}{\epsilon}$.

Before we state the main result of this section, we introduce the notion of an approximate solution to \eqref{agg_opt} in the following definition.

\begin{definition}
A solution $\br{\alpha,f,\lambda^-,\lambda^+,\mu^-,\mu^+,\nu}$ to \eqref{agg_opt} is an \textbf{$\epsilon$-accurate solution} if the constraints are violated by at most $\epsilon$ and $\gamma\br{\alpha}\geq \gamma^*-\epsilon$.
\end{definition}

Note that, if one is simply interested in approximating $\gamma^*$ (as a market regulator would be), the $\epsilon$-constraint violation is of no consequence, and an $\epsilon$-accurate solution of \eqref{agg_opt} suffices to compute an $\epsilon$-approximation to $\gamma^*$.

Given the above notion of approximation, our main theorem is the following (proof is in Appendix \ref{sec:App}):
\begin{theorem}\label{thm:Main}
An $\epsilon$-accurate solution to the optimal aggregator curtailment problem \eqref{agg_opt} for an $n$-bus radial network can be found by an algorithm with running time  $c n\powb{\frac{1}{\epsilon}}{9}$ where $c$ is a constant that depends on the parameters $p_i^a,B,d,p,\underline{f},\overline{f}$. On a linear (feeder line) network, the running time reduces to $c n\powb{\frac{1}{\epsilon}}{6}$.
\end{theorem}

We now give an informal description of the approximation algorithm. Consider a radial distribution network with nodes labeled $i \in [n]$, (where $1$ denotes the substation bus, where the radial network connects to the transmission grid). Radial distribution networks have a \emph{tree} topology (they do not have cycles). We denote bus $1$ as the \emph{root} of the tree, and buses with only one neighbor as \emph{leaves}. Every node (except the root) has a unique \emph{parent}, defined as the first node on the unique path connecting it to the root node. The set of nodes $k$ that have a give node $i$ as its parent are said to be its \emph{children}. It can be shown that the strategic curtailment problem on any radial distribution network can be expressed as an equivalent problem on a network where each node has maximum degree $3$ (known as a \emph{binary tree}, see Appendix \ref{sec:App}). Thus, we can limit our attention to networks of this type, where every node has a unique parent and at most $2$ children.

For a node $i$, let $c_1\br{i},c_2\br{i}$ denote its children (where $c_1=\emptyset$, $c_2=\emptyset$ is allowed since a node can have fewer than two children). We use the shorthand
\[p^{net}\br{i}= f_{c_1\br{i}}+f_{c_2\br{i}}-f_{i}-\br{p_i-\alpha_i-d_i}\]
\eqref{eq:DCPF} reduces to $\underline{\Delta p}_i \leq p^{net}\br{i}\leq \overline{\Delta p}_i$
where $f_1 = 0$ and $f_\emptyset = 0$. The matrix $H$ in \eqref{eqn: Hf=0} is an empty matrix (the nullspace of the matrix $B$ is of dimension $0$), so this constraint can be dropped. Using this additional structure, the problem \eqref{agg_opt} can be rewritten (after some algebra) as:
\begin{subequations}\label{opt:aggrTree}
\begin{align}
& \underset{\lambda,f,\alpha}{\text{maximize}}\qquad \sum\limits_{i=1}^n \lambda_i\br{p_i^a-\alpha_i} \\
& \text{subject to}\notag \\
&0\leq \alpha_i\leq p_i^a, \quad i\in [n] \\
& \dpl_i \leq p^{net}\br{i} \leq \dpu_i, \quad i\in [n]\\
&\fl_i \leq f_i \leq \fu_i, \quad i\in [n]\setminus \{1\}\\
&
\lambda_i \quad \begin{cases}\leq c_i, & \text{ if } p^{net}\br{i}=\dpl_i\\
					= c_i, & \text{ if } \dpl_i<p^{net}\br{i}<\dpu_i\\\geq c_i, & \text{ if } p^{net}\br{i}=\dpu_i\end{cases} ,i \in [n] \\
& \lambda_{c_j\br{i}} -\lambda_{i} \begin{cases}\geq 0, & \text{ if } f_{i}=\fl_{i}\\= 0, & \text{ if } \fl_{i} < f_{i} < \fu_{i}\\\leq 0, & \text{ if } f_{i}=\fu_{i}\end{cases} ,i\in [n],j=1,2
\end{align}
\end{subequations}
where $\lambda_i$ is the LMP at bus $i$. Note that we assumed that there is some aggregator generation and potential curtailment at every bus (however this is not restrictive, since we can simply set $p_i^a=0$ at buses where the aggregator owns no assets).

Define $x_i = (\lambda_i, f_i, \alpha_i)$. \eqref{opt:aggrTree} is of the form:\\
$\underset{x:h_i\br{x_i,x_{c_1\br{i}},x_{c_2\br{i}}}\leq 0,i\in [n]  }{\max}  \sum\limits_{i=1}^n g_i(x_i)
$
for some functions $g_i(.)$ and $h_i(.)$. This form is amenable to dynamic programming, since if we fix the value of $x_i$, the optimization problem for the subtree under $i$ is decoupled from the rest of the network. Set $\kappa_n\br{x}=0$, define $\kappa_i$ for $i<n$ recursively as:\\
$\kappa_i\br{x} =\max_{\substack{ x_{c_1\br{i}},x_{c_2\br{i}} \\ h_i\br{x,x_{c_1\br{i}},x_{c_2\br{i}}}\leq 0}} \sum_{j=1}^2 g_{c_j\br{i}}\br{x_{c_j\br{i}}}+\kappa_{c_j\br{i}}\br{x_{c_j\br{i}}}$
Then, the optimal value can be computed as $\gamma^* = \max_{x} \kappa_1\br{x}+g_1\br{x}.$ However, the above recursion requires an infinite-dimensional computation at every step, since the value of $\kappa_i$ needs to be calculated for \emph{every} value of $x$. To get around this, we note that the variables $\lambda_i,f_i,\alpha_i$ are bounded, and hence $x_i$ can be discretized to lie in a certain set $\Xcal_i$ such that every feasible $x_i$ is at most $\delta(\epsilon_i)$ away (in infinity-norm sense) from some point in $\Xcal_i$ (Lemma \ref{lem:Disc}). The discretization error can be quantified, and this error bound can be used to relax the constraint to $h_i\br{x_i,x_{i+1}} \leq \epsilon$ guaranteeing that any solution to \eqref{agg_opt} is feasible for the relaxed constraint. This allows us to define a dynamic program (algorithm \ref{alg:DPtree}).
\begin{figure}[t]
\centering
\includegraphics[width=.25\columnwidth]{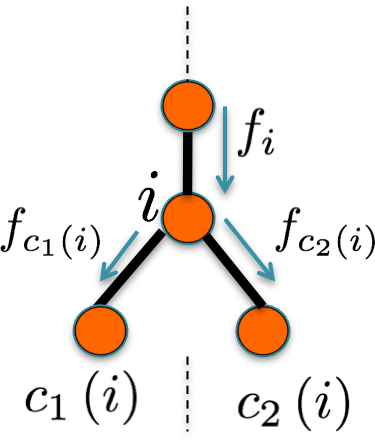}
\caption{The representation of a binary tree. For any node $i$, and its children denoted $c_1(i), c_2(i)$.}
\label{fig:tree}
\end{figure}

\begin{algorithm}[H]\caption{Dynamic programming on binary tree}\label{alg:DPtree}
\begin{center}
\begin{algorithmic}
\State $S \gets \{i: c_1\br{i}=\emptyset,c_2\br{i}=\emptyset\}$
\State $\kappa_i\br{x} \gets 0\ \forall x \in \Xcal_i,i\in S$
\While{$|S| \leq n $}
\State $S^\prime \gets \{i\not\in S: c_1\br{i},c_2\br{i}\in S\}$
\State $\forall i \in S^\prime, \forall x \in \Xcal_i$:
\[\kappa_i\br{x} \gets \max_{\substack{ x^\prime_1 \in \Xcal_{c_1(i)},x^\prime_2 \in \Xcal_{c_2(i)} \\ h_i\br{x,x^\prime_1,x^{\prime}_2}\leq \epsilon}}\sum_{j=1,2} g_{c_j\br{i}}\br{x^\prime_j}+\kappa_{c_j(i)}\br{x^\prime} \]
\State $S\gets S \cup S^{\prime}$
\EndWhile
\State $\gamma \gets \max_{x \in \Xcal_1} \kappa_1\br{x}+g_1\br{x}$
\end{algorithmic}
\end{center}
\end{algorithm}	
The algorithm essentially starts at the leaves of the tree and proceeds towards the root, at each stage updating $\kappa$ for nodes whose children have already been updated (stopping at root). Along with the discretization error analysis in Appendix \ref{sec:App}, this essentially concludes Theorem~\ref{thm:Main}.

It is worth noting that previous work on distribution level markets have used AC power flow models (at least in some approximate form) due to the importance of voltage constraints and reactive power in a distribution system \cite{ntakou2014distribution}. Our approach extends in a straightforward way to this setting as well, as the dynamic programming structure remains preserved (the KKT conditions will simply be replaced by the corresponding conditions for the AC based market clearing mechanism).

\section{The Impact of Strategic Curtailment}
\label{sec:experiment}

Using the algorithms developed in the previous sections, we can now move to characterizing the potential impact of strategic curtailment. We do this by first providing an illustrative example of how curtailment leads to a larger profit for a simple single-bus aggregator in a small, 6-bus network. Then, we consider curtailment in larger networks and show the global effect of strategic behavior. We use IEEE 14-, 30-, and 57-bus test cases, and their enhanced versions from NICTA Energy System Test case Archive \cite{coffrin2014nesta}.

\begin{figure}[htb]
\centering
\includegraphics[width=.7\columnwidth]{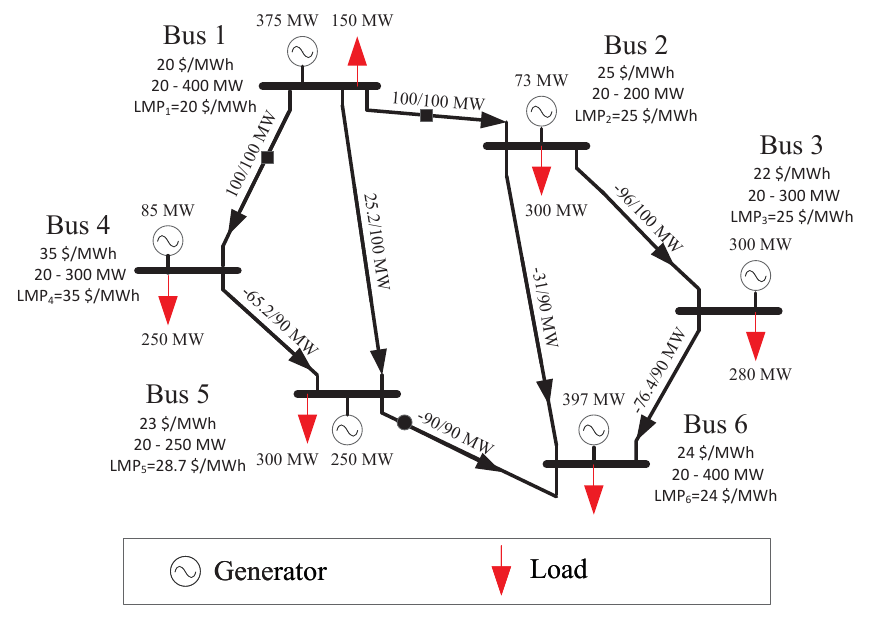}
\caption{The 6-bus example network from \cite{bi2013false}, used to illustrate the effect of curtailment.}
\label{6-bus}
\end{figure}

\subsection{An Illustrative Example}

\begin{figure}[htb]
\centering
\includegraphics[width=2.3in]{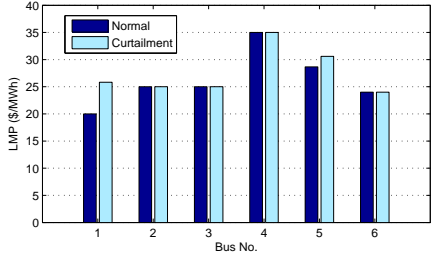}
\caption{The locational marginal prices for the 6-bus example before and after the curtailment.}
\label{bar}
\end{figure}

We use a 6-bus example network from \cite{bi2013false}, in which the amounts of generation are 375.20, 73.00, 299.60, 84.80, 250.00, 397.40 $MW$ (See Fig. \ref{6-bus}). The loads and the original offer prices for the generators can be found in the figure. In this case, 
the lines $l_{12}$, $l_{14}$ and $l_{56}$ are carrying their maximum flow,
and the real-time LMPs are 20.0, 25.0, 25.0, 35.0, 28.7, 24.0 $ \$/{MWh}$, respectively.

Assume that the aggregator owns node 1 and aims to increase its profit by curtailing the generation at this node. It can be seen that by decreasing the generation at node 1 by 0.15, from 375.20 $MW$ to 375.05 $MW$, 
the binding and non-binding constraints in problem \eqref{ISO_opt} change, and as a result the ISO determines the new LMPs as 25.8, 25.0, 25.0, 35.0, 30.6, 24.0 $ \$/{MWh}$. Fig. \ref{bar} shows the LMPs, before and after the curtailment. 
The curtailment profit is
$\gamma =25.8\times 375.05-20\times 375.20$ 
$ = 2172\ ^{\$}/_h$ ,
which means that the aggregator has been able to increase its profit by 2172 $ \$/h$. For dispatch intervals of length 15 minutes, this amount is around \$543 per dispatch interval. 
\subsection{Case Studies}
\begin{figure}[htb]
\centering
\includegraphics[width=0.5\columnwidth]{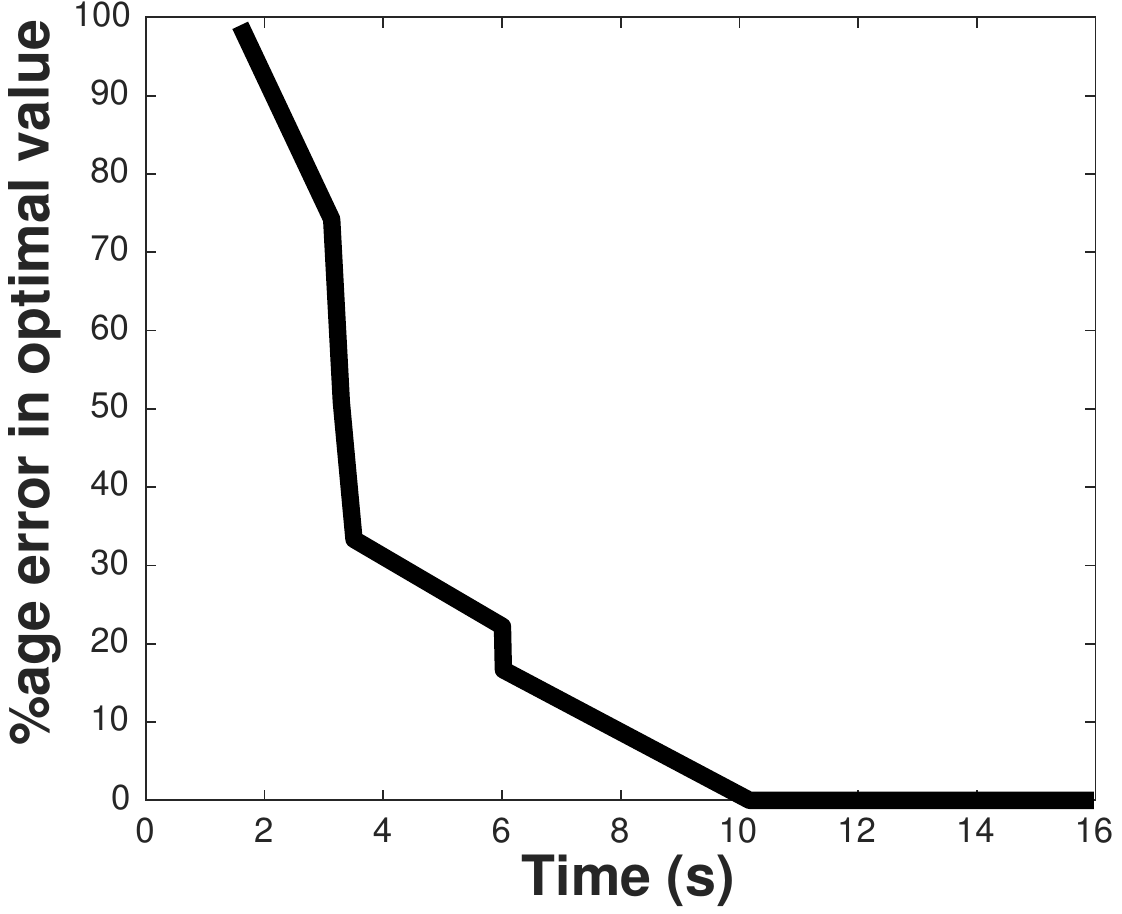}
\caption{The difference from the optimal solution as a function of the running time of the algorithm, on a 9-bus network with 1\% curtailment allowance.}
\label{run-time}
\end{figure}
\begin{figure*}[!t]
\centering
\subfloat{\includegraphics[width=0.5\columnwidth]{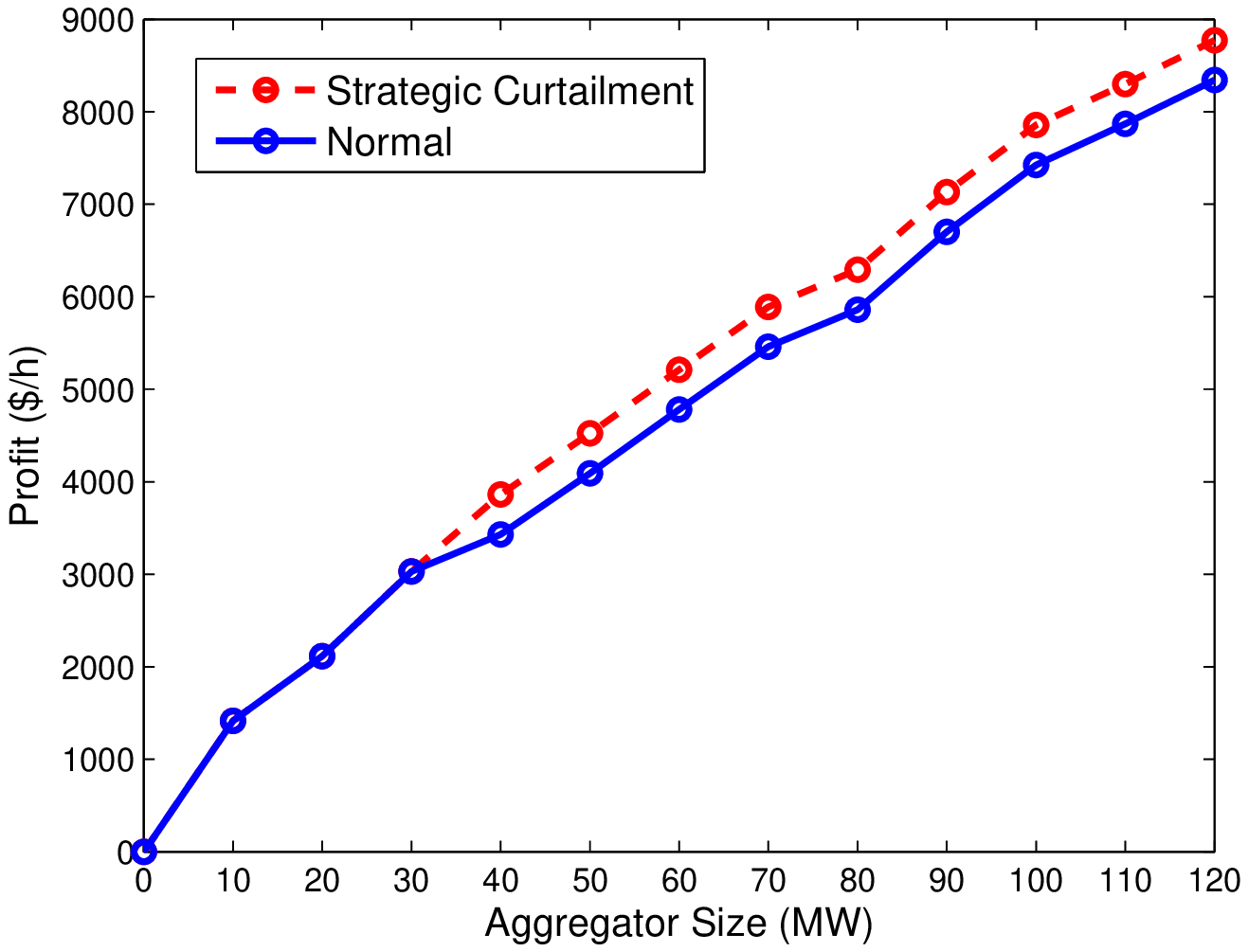}%
\label{profit14}}
\subfloat{\includegraphics[width=0.5\columnwidth]{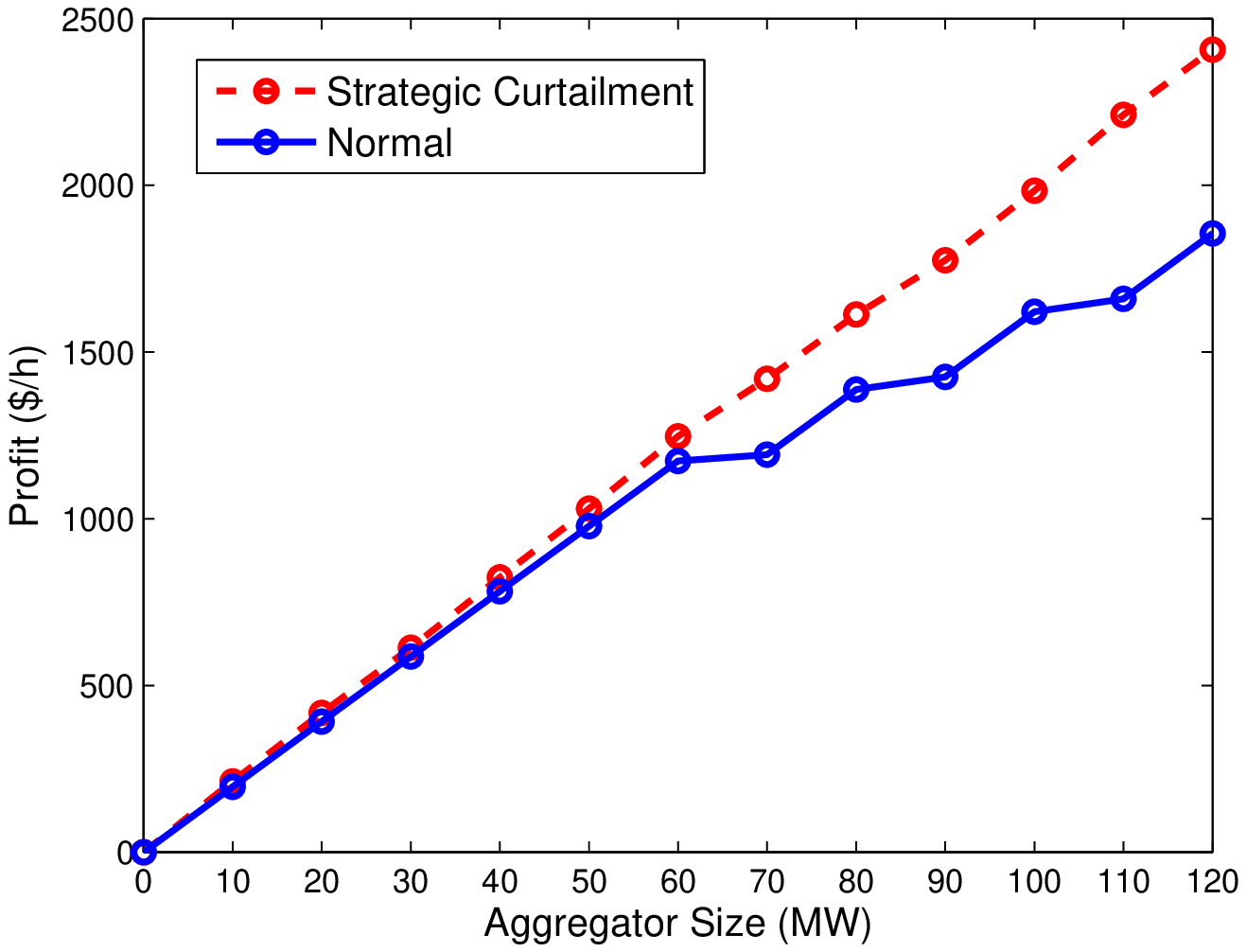}%
\label{profit30}}
\subfloat{\includegraphics[width=0.5\columnwidth]{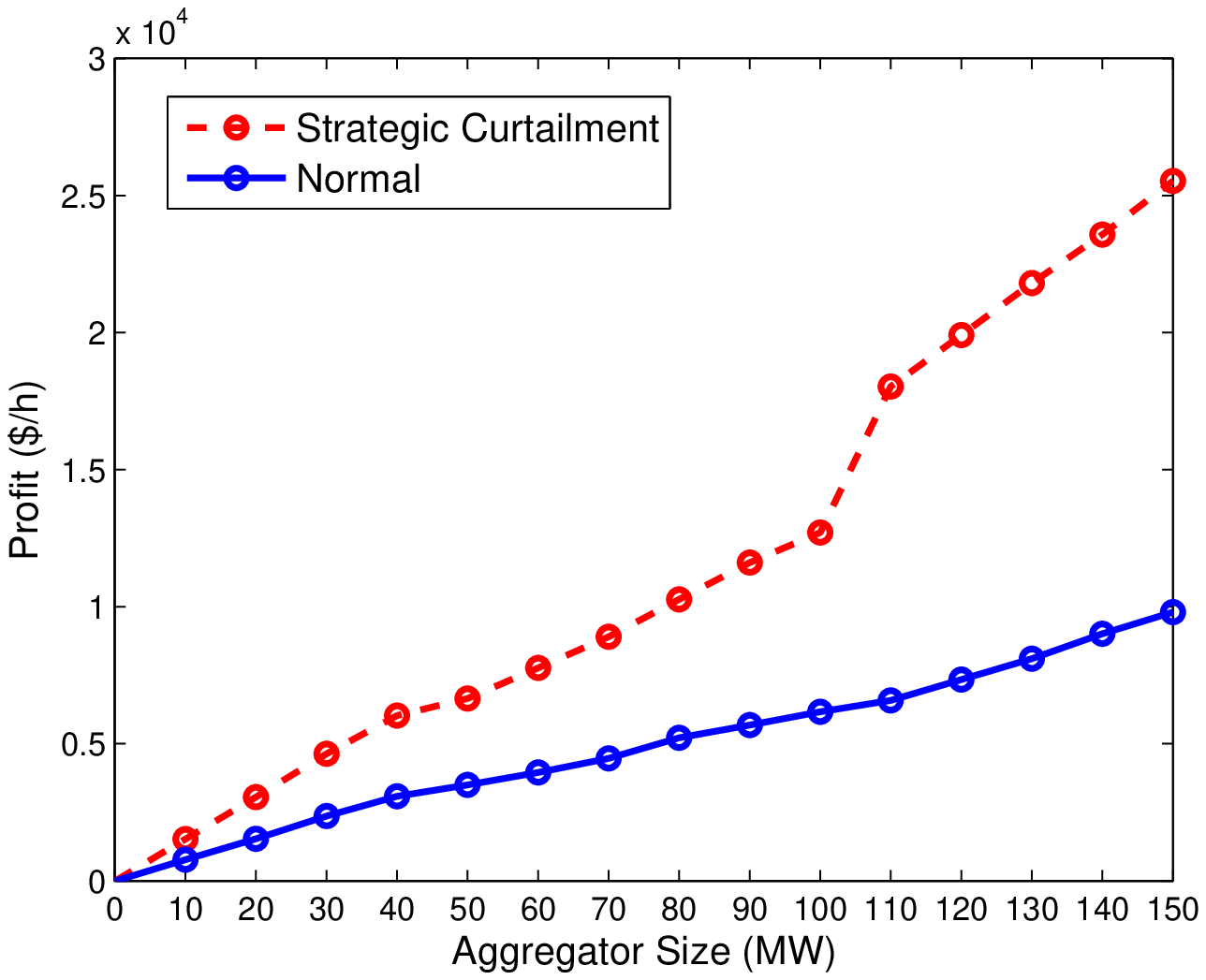}%
\label{profit57}}
\caption{The profit under the normal (no-curtailment) condition and under (optimal) strategic curtailment, as a function of size of the aggregator in IEEE test case networks: a) IEEE 14-Bus Case, b) IEEE 30-Bus Case, and c) IEEE 57-Bus Case. The difference between the blue and red curves, which is the curtailment profit, is non-decreasing.}
\label{fig_profits}
\end{figure*}

We simulate the behavior of aggregators with different sizes, i.e. different number of buses, in a number of different networks. We use the IEEE 14-, 30-, and 57-bus test cases. Since studying market manipulation makes sense only when there is congestion in the network, we scale the demand (or equivalently the line flows) until there is some congestion in the network. In order to examine the profit and market power of aggregator as a function of its size, we assume that the way aggregator grows is by sequentially adding random buses to its set (more or less like the way e.g. a solar firm grows). Then at any fixed set of buses, it can choose different curtailment strategies to maximize its profit. In other words, for each of its nodes it should decide whether to curtail or not (assuming that the amount of curtailment has been fixed to a small portion). We assume that the total generation of the aggregator in each bus is 10 $MW$ and it is able to curtail 1\% of it (0.1 $MW$).

For each of the three networks, Fig. \ref{fig_profits} shows the profit for a random sequence of nodes. Comparing the no-curtailment profit with the strategic-curtailment profit, reveals an interesting phenomenon. As the size of the aggregator (number of its buses) grows, not only does the profit increase (which is expected), but also the difference between the two curves increase, which is the ``curtailment profit.'' More specifically, the latter does not need to happen in theory. However in practice, it is observed most of the time, and it points out that larger aggregators have higher incentive to behave strategically, and they can indeed gain more from curtailment.


Finally, to demonstrate the performance of our approximation algorithm on acyclic networks, we plot the suboptimality gap versus running time of the algorithm (Fig. \ref{run-time}) on a radial version of the IEEE 9-bus network (taken from \cite{AndyKocuk} ). It can be seen that with the optimal solution can be computed with $0$ error in $10s$.

\section{Concluding Remarks}
Understanding the potential for market manipulation by aggregators is crucial for electricity market efficiency in the new era of renewable energy. In this paper, we characterized the profit an aggregator can make by strategically curtailing generation in the ex-post market as the outcome of a bi-level optimization problem. This model captures the realistic price clearing mechanism in the electricity market. With this formulation when the aggregator is located in a single bus, we have shown that the locational marginal price is monotonically increasing with the curtailment; the profit directly correlates with market power defined in the literature, and we have an exact polynomial-time algorithm to solve the aggregators profit maximization problem.

The aggregator's strategic curtailment problem in a general setting is a difficult bi-level optimization problem, which is intractable. However, we show that for radial distribution networks (where aggregators are likely located) there is an efficient algorithm to approximate the solution up to arbitrary precision. Finally, we show via simulations on realistic test cases that, first, there is potentially large profit for aggregators by manipulating the LMPs in the electricity market, and second, our algorithm can efficiently find the (approximately) optimal curtailment strategy.

We view this paper as a first step in understanding market power of aggregators, and more generally, towards market design for integrating renewable energy and demand response from geographically distributed sources. With the result of our paper, it is interesting to ask what can the operator do to address this problem. In particular, how to design market rules for aggregators to maximize the contribution of renewable energy yet mitigate the exercise of market power. Also, extending the analysis to the case of multiple aggregators in the market is another interesting direction for future research.

\bibliographystyle{IEEEtran}
\bibliography{references}

\appendix 
\subsection{Proof of Lemma~\ref{lem: mono}}\label{sec:monoproof}
Let us take a look at the ISO's optimization problem \eqref{ISO_opt}, which is a linear program. It is not hard to see that the dual of this problem is as follows.
\begin{subequations}\label{ISO_dual}
\begin{align}
& \underset{\boldsymbol{\lambda}^-,\boldsymbol{\lambda}^+,\boldsymbol{\mu}^-,\boldsymbol{\mu}^+,\boldsymbol{\nu}}{\text{maximize}} \quad (\mathbf{\dpl+p-\boldsymbol{\alpha}-d})^T\boldsymbol{\lambda^-}+\notag\\
& \qquad \qquad (\mathbf{-p+\boldsymbol{\alpha}+d-\dpu})^T\boldsymbol{\lambda^+}+\mathbf{\fl}^T\boldsymbol{\mu^-}-\mathbf{\fu}^T\boldsymbol{\mu^+}\\
&  \text{subject to} \quad\notag\\
& \mathbf{B}^T(\boldsymbol{c+\lambda^+-\lambda^-})-\boldsymbol{\mu}^-+\boldsymbol{\mu}^++\mathbf{H}^T\boldsymbol{\nu} = \mathbf{0}\\
& \boldsymbol{\lambda}^-,\boldsymbol{\lambda}^+,\boldsymbol{\mu}^-,\boldsymbol{\mu}^+\geq 0
\end{align}
\end{subequations}

If one focuses on the terms involving $\alpha_i$ for a certain $i$, the objective of the above optimization problem is in the form: $(\dpl_i+p_i-\alpha_i-d_i)\lambda^-_i+(-p_i+\alpha_i+d_i-\dpu_i)\lambda^+_i$ plus a linear function of the rest of the variables (i.e. the rest of $\boldsymbol{\lambda}^-,\boldsymbol{\lambda}^+$, as well as $\boldsymbol{\mu}^-,\boldsymbol{\mu}^+,\boldsymbol{\nu}$). There is no $\boldsymbol{\alpha}$ in the constraints, and the first two terms of this objective are the only parts where $\alpha_i$ appears (and with opposite signs).

We need to show that if $\alpha_i$ is changed to $\alpha_i+\delta$ for some $\delta>0$, then $c_i+\lambda_i^{+new}-\lambda_i^{-new}\geq c_i+\lambda_i^+-\lambda_i^-$, where $\lambda_i^{+new}, \lambda_i^{-new}$ are the optimal solutions of the new problem. 

We prove this in a general setting. Consider the following two optimization problems.
\begin{subequations}\label{opt1}
\begin{align}
f^* = & \underset{\substack{x_1,x_2\in\mathbb{R}\\ x_3\in\mathbb{R}^m}}{\sup} && a_1x_1+a_2x_2+a_3^Tx_3\\
& \quad \text{ s.t.} && (x_1,x_2,x_3)\in S
\end{align}
\end{subequations}
\begin{subequations}\label{opt2}
\begin{align}
f^{*new} = & \underset{\substack{x_1,x_2\in\mathbb{R}\\ x_3\in\mathbb{R}^m}}{\sup} && (a_1-\delta)x_1+(a_2+\delta)x_2+a_3^Tx_3\\
& \quad \text{ s.t.} && (x_1,x_2,x_3)\in S
\end{align}
\end{subequations}
Assume that the optimal values of the problems are attained at $(x_1^*,x_2^*,x_3^*)$ and $(x_1^{*new},x_2^{*new},x_3^{*new})$, respectively.

We claim that $x_2^{*new}-x_1^{*new}\geq x_2^*-x_1^*$ (This precisely implies the LMP condition in our case, i.e. $\lambda_i^{+new}-\lambda_i^{-new}\geq \lambda_i^+-\lambda_i^-$).

Suppose by way of contradiction that $x_2^{*new}-x_1^{*new}<x_2^*-x_1^*$.\\
We know that $a_1x_1^*+a_2x_2^*+a_3^Tx_3^*\geq a_1x_1+a_2x_2+a_3^Tx_3, \ \forall (x_1,x_2,x_3)\in S$.\\
Therefore we have
\begin{align*}
&(a_1-\delta)x_1^{*new}+(a_2+\delta)x_2^{*new}+a_3^Tx_3^{*new}\\
&=a_1x_1^{*new}+a_2x_2^{*new}+a_3^Tx_3^{*new}-\delta x_1^{*new}+\delta x_2^{*new}\\
&\leq a_1x_1^*+a_2x_2^*+a_3^Tx_3^*+\delta(x_2^{*new}-x_1^{*new})\\
&< a_1x_1^*+a_2x_2^*+a_3^Tx_3^*+\delta(x_2^*-x_1^*)\\
&=(a_1-\delta)x_1^*+(a_2+\delta)x_2^*+a_3^Tx_3^* .
\end{align*}
The first inequality above follows from the fact that $(x_1^{*new},x_2^{*new},x_3^{*new})\in S$. Now the above implies that $(x_1^{*new},x_2^{*new},x_3^{*new})$ is not the optimal solution of \eqref{opt2}, and it is a contradiction.

As a result, $x_2^{*new}-x_1^{*new}\geq x_2^*-x_1^*$. \qed


\subsection{Proof of Theorem~\ref{thm: single_bus_opt}}\label{sec:exactproof}
Since we are in the single-bus curtailment regime, $\alpha$ has only one nonzero component. For the sake of convenience, we denote that element itself by a scalar $\alpha$ throughout this proof (no $\alpha$ is vector in this proof). The proof consists of the following two pieces: 1) From each jump point, the point where the next jump happens can be computed in polynomial time, 2) There are at most polynomially (in this case even linearly) many jumps.

Assuming that the solution to the program \eqref{ISO_opt} is unique, for any fixed value of $\alpha$, exactly $t$ of the constraints (\ref{const: p},\ref{const: f},\ref{const: range}) are binding (active). We can express these binding constraints as
$$Af=b(\alpha) ,$$
where $A\in\mathbb{R}^{t\times t}$ is an invertible matrix, and $b(\alpha)\in\mathbb{R}^t$ is a vector that depends on $\alpha$.
As long as the binding constraints do not change, the matrix $A$ is fixed and the optimal solution is linear in $\alpha$ (i.e. $f=A^{-1}b(\alpha)$).
Then, for simplicity, we can express the solution as $f(\alpha)=f_0+\alpha a$, for some $t$-vectors $f_0$ and $a$.

Now if we look at the non-binding (inactive) constraints of \eqref{ISO_opt}, they can also be expressed as
$$\tilde{A}f< \tilde{b} ,$$
for some matrix $\tilde{A}$ and vector $\tilde{b}$ of appropriate dimensions. Inserting $f$ into this set of inequalities yields $\tilde{A}f_0+\alpha\tilde{A}a<\tilde{b}$, or equivalently
$$\alpha(\tilde{A}a)_i<\tilde{b}_i-(\tilde{A}f_0)_i ,$$
for all $i=1,2,\dots,(2n+2t-rank(G))$.

Now we need to figure out that, with increasing $\alpha$, which of the non-binding constraints becomes binding first and with exactly how much increase in $\alpha$. If for some $i$ we have $(\tilde{A}a)_i\leq 0$, then it is clear that increasing $\alpha$ cannot make constraint $i$ binding. If $(\tilde{A}a)_i>0$ then the constraint can be written as
$$\alpha<\frac{\tilde{b}_i-(\tilde{A}f_0)_i}{(\tilde{A}a)_i} .$$

Computing the right-hand side for all $i$, and taking their minimum, tells us exactly which constraint will become binding next and how much change in the current value of $\alpha$ results in that.

The complexity of this procedure is $O(t^{2.373})$ for computing $f_0$ and $a$, plus $O(t(2n+2t))=O(nt+t^2)$ for computing the lowest bound among all the constraints. Hence the complexity is $O(t^{2.373})$.

The above procedure describes how the next jump point can be computed efficiently from the current point. The exact same procedure can be repeated for reaching the subsequent jump points. All remains to show is the second piece of the proof, which is that the number of jump points are bounded polynomially.
To show the last part note that by increasing $\alpha$, if a binding constraint becomes non-binding, it will not become binding again. As a result, each constraint can change at most twice, and therefore the number of jumps is at most twice the number of constraints. Thus, the number of jumps is $O(n+t)$, and the overall complexity of the algorithm is $O((n+t)t^{2.373})=O(t^{3.373})$. \qed
\subsection{Proof of Proposition~\ref{thm: marketpower}}\label{sec:marketpowerproof}
From the definition of $\gamma(\alpha^*) = \lambda_i(\alpha^*)(p_i^a-\alpha_i^*)-\lambda_i(0)p_i^a$ it follows that
\begin{align}
\frac{\gamma(\alpha^*)}{\lambda_i(0)(p_i^a-\alpha_i^*)} &= \frac{\lambda_i(\alpha^*)}{\lambda_i(0)} - \frac{p_i^a}{p_i^a-\alpha_i^*}\notag\\
&= 1+ \frac{\lambda_i(\alpha^*)-\lambda_i(0)}{\lambda_i(0)}-(1-\frac{\alpha_i^*}{p_i^a})^{-1}\notag\\
&\simeq 1+ \frac{\lambda_i(\alpha^*)-\lambda_i(0)}{\lambda_i(0)}-(1+\frac{\alpha_i^*}{p_i^a})\notag\\
&= \frac{\lambda_i(\alpha^*)-\lambda_i(0)}{\lambda_i(0)}-\frac{\alpha_i^*}{p_i^a}\label{eq: marketpower_appx} .
\end{align}
Therefore we have
\begin{align*}
\frac{p_i^a}{\lambda_i(0)(p_i^a-\alpha_i^*)\alpha_i^*}\gamma(\alpha^*) &= \left(\frac{\lambda_i(\alpha^*)-\lambda_i(0)}{\lambda_i(0)}\right)/\left(\frac{\alpha_i^*}{p_i^a}\right) - 1\\
&= \eta_i-1 .
\end{align*}
Since the left-hand side parameters are all positive, if $\gamma(\alpha^*)>0$, we can conclude that $\eta_i>1$. Moreover, it is clear that the larger the value of $\gamma(\alpha^*)$ is, the higher the value of $\eta_i$ is. Note that we used the approximation $(1-\frac{\alpha_i^*}{p_i^a})^{-1} \simeq 1+\frac{\alpha_i^*}{p_i^a}$, since the curtailment is small with respect to the generation; however, the right-hand side expression \eqref{eq: marketpower_appx} is an upper bound on the left-hand side anyway, and the result holds exactly. \qed
\subsection{Proof of Theorem~\ref{thm:Main}}\label{sec:App}

\begin{lemma}[$\delta$-discretization]\label{lem:Disc}
Given a set $C \subset [\underline{L_1},\overline{L_1}]\times\dots\times[\underline{L_k},\overline{L_k}]$, there exists a finite set $\Xcal$ such that 
\[\forall z \in C \quad \exists z^\prime \in \Xcal, \max_{1 \leq i \leq k} |z_i-z^\prime_i| \leq \delta\] 	and $\Xcal$ contains at most $V/\delta^k$ points, where $V=\prod_{i=1}^k (\overline{L_i}-\underline{L_i})$ is a constant (the volume of the box). $\Xcal$ is said to be an $\delta$-discretization of $C$ and written as $\Xcal\br{\delta}$.
\end{lemma}

\begin{lemma}[Reduction to Binary Tree]
Any tree with arbitrary degrees can be reduced to a binary tree by introducing additional dummy nodes to the network.
\end{lemma}
\begin{proof}
Take any node $b$ in the tree with some parent $a$ and $k>2$ children $c_1,\dots,c_k$.
There exists $m>0$ such that $2^m<k\leq 2^{m+1}$ for some $m$. We will show that this subgraph can be made a binary tree by introducing $O(k)$ dummy nodes (in $m$ levels) between $b$ and its children. The additional nodes and edges are defined as follows:
$$b\rightarrow b_0,\ b\rightarrow b_1,$$
$$b_0\rightarrow b_{00},\ b_0\rightarrow b_{01},\ b_1\rightarrow b_{10},\ b_1\rightarrow b_{11},$$
$$b_{00}\rightarrow b_{000},\ b_{00}\rightarrow b_{001},\ \dots,\ b_{11}\rightarrow b_{111},$$
up to $m$ levels:
$$b_{0\dots 00}\rightarrow c_1,\ b_{0\dots 00}\rightarrow c_2,\ b_{0\dots 01}\rightarrow c_3\ \dots$$

This is transparently a binary tree with $O(k)$ nodes. Each of the new nodes has zero injection, and effectively the incoming flow from its parent is just split in some way between its children. This in fact enforces the flow conservation constraint at $b$. Similar construction can be applied to any node of the tree with more than two children, until no such node exists. It can be seen that the number of nodes in the new graph is still linear in $n$.
\end{proof}

So any tree can be transformed to a binary one by the above procedure. For the rest of the analysis we focus on the $\epsilon$-approximation of the dynamic program on the resulting binary tree. The optimization problem \eqref{agg_opt} on a binary tree, can be written after some algebra as the following.
\begin{subequations}
\begin{flalign}
& \underset{\boldsymbol{\lambda},\mathbf{f},\mathbf{\alpha}}{\max}\qquad \sum\limits_{i=1}^n \lambda_i (p_i-\alpha_i)&\\
& \text{subject to}\notag
\end{flalign}
\begin{equation}
0\leq \alpha_i\leq p_i^a, \quad i=1,\dots,n
\end{equation}
\begin{multline}
\dpl_i \leq f_{c_1(i)}+f_{c_2(i)}-f_i-p_i+\alpha_i+d_i \leq \dpu_i ,\\
\quad i=1,\dots,n
\end{multline}
\begin{equation}
\fl_i \leq f_i \leq \fu_i, \quad i=2,\dots,n
\end{equation}
\begin{multline}
\begin{cases}
(\lambda_i-c_i)(f_{c_1(i)}+f_{c_2(i)}-f_{i}-p_i+\alpha_i+d_i-\dpl_i)\geq 0\\
(\lambda_i-c_i)(f_{c_1(i)}+f_{c_2(i)}-f_{i}-p_i+\alpha_i+d_i-\dpu_i)\geq 0
\end{cases} ,\\
i = 1, \dots, n
\end{multline}
\begin{multline}
\begin{cases}
(\lambda_{i}-\lambda_{c_j(i)})(\fl_{c_j(i)}-f_{c_j(i)})\geq 0\\
(\lambda_{i}-\lambda_{c_j(i)})(\fu_{c_j(i)}-f_{c_j(i)})\geq 0
\end{cases} ,\\
i = 1, \dots, n, \quad j = 1,2
\end{multline}
\end{subequations}

The constraints $0\leq \alpha_i\leq p_i^a$ and $\fl_i\leq f_i\leq \fu_i$, along with a prior bound on lambda $\underline{\lambda}\leq \lambda\leq \overline{\lambda}$ can be used to define the box where $x_i=(\lambda_i,f_i,\alpha_i)$ lives. Then an $\epsilon$-accurate solution is a solution to the following problem.
\begin{subequations}\label{opt:rounded2}
\begin{flalign}
& \underset{\boldsymbol{\lambda},\mathbf{f},\mathbf{\alpha}}{\max}\qquad \sum\limits_{i=1}^n \lambda_i (p_i-\alpha_i)&\\
& \text{subject to}\notag
\end{flalign}
\begin{multline}
\dpl_i-\epsilon \leq f_{c_1(i)}+f_{c_2(i)}-f_i-p_i+\alpha_i+d_i \leq \dpu_i+\epsilon ,\\
\quad i=1,\dots,n
\end{multline}
\begin{multline}\label{constraint_example}
\begin{cases}
(\lambda_i-c_i)(f_{c_1(i)}+f_{c_2(i)}-f_{i}-p_i+\alpha_i+d_i-\dpl_i)\geq -\epsilon\\
(\lambda_i-c_i)(f_{c_1(i)}+f_{c_2(i)}-f_{i}-p_i+\alpha_i+d_i-\dpu_i)\geq -\epsilon
\end{cases} ,\\
i = 1, \dots, n
\end{multline}
\begin{multline}
\begin{cases}
(\lambda_{i}-\lambda_{c_j(i)})(\fl_{c_j(i)}-f_{c_j(i)})\geq -\epsilon\\
(\lambda_{i}-\lambda_{c_j(i)})(\fu_{c_j(i)}-f_{c_j(i)})\geq -\epsilon
\end{cases} ,\\
i = 1, \dots, n, \quad j = 1,2
\end{multline}

Assuming a $\delta$-discretization of the constraint set, each of the constraints (as well as $\epsilon$-accuracy of the objective) imposes a bound on the value of $\delta$. For example constraint \eqref{constraint_example} requires $4\delta\leq\epsilon$ (Note that we could have defined different deltas $\delta^{\lambda},\delta^f,\delta^{\alpha}$ for different variables and in that case we had $3\delta^f+\delta^p\leq\epsilon$, but for simplicity we take all the deltas to be the same). Similar bounds on $\delta$ can be obtained from the other constraints, and taking the lowest upper-bound, implies the existence of a constant $c'$ (that depends on the parameters) such that $\delta\leq\epsilon/c'$.

As a result we have a $\delta$-discretization with $|\Xcal|=V/\delta^3=c'^3V/\epsilon^3$ number of points, for any node. Therefore, the computational complexity over any node will be $|\Xcal|^3$, because we have $|\Xcal|$ many values for the node itself and $|\Xcal|$ many values for any of its two children.
Since there are $n$ nodes, the overall complexity of the algorithm will simply be $n|\Xcal|^3=nc'^9V^3/\epsilon^9=cn/\epsilon^9$. \qed

\end{subequations}

\ifCLASSOPTIONcaptionsoff
  \newpage
\fi
\end{document}